\documentclass[12pt]{article}

\usepackage{graphicx}

\usepackage{amsthm,amsfonts,amsmath}
\swapnumbers
\font\notefont=cmsl8

\theoremstyle{plain}

\newtheorem{thm}{THEOREM}[section]
\newtheorem{lm}[thm]{LEMMA}

\newtheorem{remark}[thm]{REMARK}
\theoremstyle{definition}

\theoremstyle{remark}
\newcommand{\upchi}{\raise1pt\hbox{$\chi$}}
\newcommand{\R}{{\mathord{\mathbb R}}}

\begin{document}

\title{\bf{Rigorous results for the speed of Kolmogorov--Petrovskii--Piscounov fronts with a cutoff}}
\author{\vspace{5pt} Rafael D. Benguria$^1$,
M. Cristina Depassier$^{2}$ and Michael Loss$^3$ \\
\vspace{-4pt}\small{$1.$ Departamento de F\'\i sica, Pontificia. Universidad Cat\'olica de Chile,} \\[-5pt]
\small{Casilla 306, Santiago 22, Chile, email: rbenguri@fis.puc.cl}\\
\vspace{-4pt}\small{$2.$ Departamento de F\'\i sica, Pontificia Universidad Cat\'olica de Chile,} \\[-5pt]
\small{Casilla 306, Santiago 22, Chile, email:mcdepass@gmail.com}\\
\vspace{-4pt}\small{$3.$ School of Mathematics, Georgia Tech,
Atlanta, GA 30332}\\ [-5pt]
\small{email: loss@math.gatech.edu}\\}
\date{\today}
\maketitle

\footnotetext
[1]{Work partially supported by Fondecyt (CHILE) project 1100679,
and by ICM (Chile) project P07--027--F}
\footnotetext
[3]{Work partially supported by NSF-grant DMS-0901304. \\
\copyright\, 2010 by the authors. This paper may be reproduced, in its
entirety, for non-commercial purposes.}

\begin{abstract}
We study the effect of a cut-off on the speed of pulled fronts of the one dimensional reaction diffusion equation. We prove rigorous upper and lower bounds on the speed in terms of the cut-off parameter $\epsilon$. From these bounds we estimate the range of validity of the Brunet--Derrida formula for a general class of reaction terms. \end{abstract}

\section{Introduction}

The reaction diffusion equation
\begin{equation}
u_t = u_{xx} + f(u)
\label{rd}
\end{equation}
 is one of the simplest models which shows  how a small perturbation to an unstable state develops into a moving front joining a stable to an unstable state. The reaction term $f(u)$ adopts different expressions depending on the physical problem under consideration. One of the most studied cases, is the Fisher reaction term \cite{Fisher}
$f(u) = u (1 -u)$ for which  the asymptotic speed of the propagating
front is $c=2$, a value determined from linear considerations.
 A more general case was studied by Kolmogorov, Petrovskii and Piscounov (KPP)\cite{KPP} who showed that for all reaction terms which satisfy the so called KPP condition
\begin{equation}
f(u) > 0, \qquad f(0) = f(1) = 0, \qquad  f(u) < f'(0) u
\label{KPPcondition}
\end{equation}
the asymptotic speed of the front joining the stable $u=1$ point  to the unstable $u=0$ point is given by
$$
c_{KPP} = 2 \sqrt{f'(0)}.
$$
The evolution of localized initial conditions to the front of minimal speed was established in \cite{AW78} for general reaction terms. Recent work has dealt  with effects not included in the classical reaction diffusion equation (\ref{rd}), namely the effects of noise and of the finiteness in the  number $N$ of diffusive particles. It was suggested  by Brunet and Derrida \cite{BrunetDerrida97} that such effects can be simulated by introducing a cut-off in the reaction term. In the case of noise the cut-off parameter measures the amplitude of the noise while in the case of finite number of $N$ diffusing particles the cut-off parameter $\epsilon = 1/N$. They presented numerical evidence to support their conjecture.
By means of an asymptotic matching Brunet and Derrida  showed that for a reaction term $f(u) = u(1 - u^2)$ a small cut-off changes the speed of the front to
\begin{equation}
c \approx 2 -\frac{\pi^2}{(\log \epsilon)^2}.
\label{brunet}
\end{equation}
In recent work it has been show that the Brunet-Derrida formula for the speed is correct to $\cal{O}$$((\log \epsilon)^{-3})$ for a wider class of pulled reaction terms and cut-off functions. \cite{BDL08,DPK07}. A completely different behavior is found when a cut-off is applied to a bistable reaction term or to a pushed front, in these two cases the cut-off changes the speed by an amount which has   a power law dependence of the cut-off parameter \cite{BDH07,DPK10}.
The validity of representing the finiteness in the number of particles in the diffusion process by a reaction diffusion equation with a cutoff, and the effect of noise in the reaction diffusion equation with a cut--off was proved rigorously in  \cite{Berard2009B} and \cite{MMQ08}, respectively.

The purpose of this work is to prove rigorous upper and lower bounds for the speed of fronts  for reaction terms  of the form $f(u) \Theta(u-\epsilon)$ where $f$ satisfies the KPP condition Eq. (\ref{KPPcondition}) and $\Theta$ is the step function. The results obtained are valid for all $\epsilon$, in the limit of $\epsilon \rightarrow 0$ the upper and lower limits coincide and are the Brunet--Derrida value.

The paper is organized as follows: in Section 2 we set up the problem and state our main result (Theorem \ref{main}). In Section 3 we consider
a relaxed variational principle which is crucial in determining upper bounds on the speed of propagation of fronts. In particular we prove existence 
of a unique minimizer for the relaxed variational problem. In Section 4, we give an explicit expression for the minimizer of the relaxed problem. At the same time we characterize (in closed form) the value of the minimum of the relaxed problem, which allows us to give an explicit upper bound for the speed of propagation of fronts for ({\ref{rd}). In Section 5, we prove our main result (i.e., Theorem \ref{main}) and, in particular we provide error bounds for $c$.

\section{Statement of the problem}

The bounds will be obtained starting from the integral variational principle for the speed of the fronts. It was shown in previous work that the speed of the front of the reaction diffusion equation \ref{rd} with arbitrary reaction term $f(u)$ obeys the variational principle \cite{BD96c,BDPRE07}
\begin{equation}
c^2 = \sup_{u(s)} 2\, \frac{ F(1)/s_0 + \int_0^{s_0} F(u(s))/s^2 d\,s}{\int_0^{s_0} \left( d u /d s\right)^2 d\,s},
\label{newvp}
\end{equation}
where $s_0= 1/g(u=1)$ is an arbitrary parameter,
$$
F(u) = \int_0^u f(q) dq.
$$
and  the supremum is taken over positive increasing functions $u(s)$ such that $u(0)=0$, $u(s_0) =1$ and for which all the integrals in (\ref{newvp}) are finite.

We shall now be interested in
 reaction terms $f(u)$ with a cut-off $\epsilon$ of the form
$$
f(u) = \left\{ \begin{array}{ll}
              0
             & \mbox{if $0\leq u \leq \epsilon$} \\
             u - N(u) &\mbox{if $\epsilon < u < 1,$}
           \end{array} \right.
\label{KPPprofile}
$$
where $N(u)$, the nonlinearity, is such that $N(0)= N'(0) =0$. Moreover we will assume that $f(u)$ satisfies the KPP criterion Eq.(\ref{KPPcondition}), which implies $N(u) \ge 0$.

Our main result is the following theorem. 

\begin{thm} \label{main}
Consider the reaction diffusion equation (\ref{rd}) where the reaction profile satisfies (\ref{KPPprofile}). 
Moreover, assume $N(u) \le B (u-\epsilon)^{1+\eta}$, for $\epsilon \le u \le 1$, where $\eta >0$. 
Then, the speed of propagation of fronts of the reaction diffusion equation (\ref{rd}), $c$, satisfies, 
\begin{equation}
0 \le c^2 - c_L^2 \le {\rm o} \left( \frac{1}{|\log \epsilon|^2} \right).
\end{equation}
Here, $c_L$ is given explicitly by
\begin{equation}
c_L \equiv 2 \sin \phi_*, 
\end{equation}
where $\phi_*$ is the first positive solution of the equation,
\begin{equation}
\phi_* \tan \phi_* = \frac{1}{2} | \log \epsilon |. 
\label{eq:thm1}
\end{equation}
In particular, for $\epsilon \to 0$, we have, 
$$
c_L = 2 - \frac{\pi^2}{|\log \epsilon |^2} 
+ \rm{o} \left(\frac{1}{|\log \epsilon |^2} \right)
$$
\end{thm}

\bigskip
\noindent
{\it Remarks:} 

\noindent
i) The speed $c_L$ is the precise speed of propagation of fronts for a profile defined piecewise as
$f_L=0$, for $0 \le u \le \epsilon$, $f_L(u)=u$, for $\epsilon<u<1$, and $f_L(1)=0$ (illustrated in the figure below).

\begin{figure}[ht]\label{fig:f1.1}
\begin{center}
\includegraphics[scale=0.35]{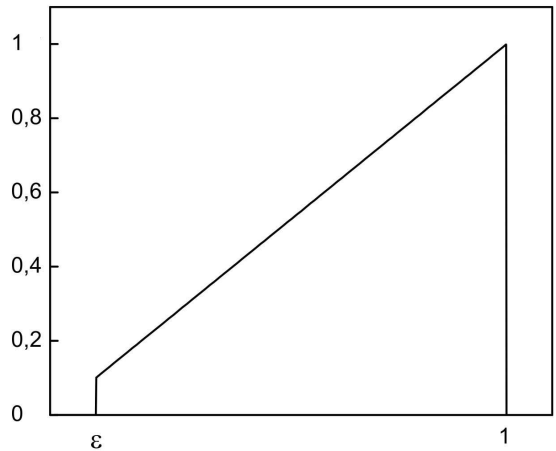}
\end{center}
\caption{the profile $f_L$ corresponding to the relaxed variational problem}
\end{figure}

\noindent
ii) Although in principle the interest is focused on small values of the parameter $\epsilon$, our expression
for $c_L = 2 \sin \phi_*$ is valid for any $0 \le \epsilon <1$. In fact, one can also consider the interesting 
case $\epsilon \to 1$. In that case, the profile $f_L$ is peaked around $u=1$, which is the typical situation 
that arises in the propagation of flames (first studied in \cite{ZeFK38}). For the case of profiles $f(u)$ peaked around $u=1$, 
the speed of fronts is approximately given by 
$$
c_{ZFK} = \sqrt{2 \int_0^1 f(u) d u},
$$
(see, \cite{ZeFK38}; it turns out that this expression $c_{ZFK}$ for the speed of the travelling fronts is actually a lower 
bound to the actual speed $c$, see \cite{BeNi92,BD98}),
Using the ZFK expression for the profile $f_L(u)$, one has, 
$$
c_{ZFK} = \sqrt{1 - \epsilon^2} \approx \sqrt{2 (1 - \epsilon)},
$$
as $\epsilon \to 1$. On the other hand, as $\epsilon \to 1$, $|\log \epsilon|= |\log(1+(\epsilon -1)| \approx 1-\epsilon$, approaches zero. 
Using (\ref{eq:thm1}), we see that also $\tan \phi_*\approx 0$ in this case, and we have $\sin\phi_* \approx\tan \phi_* \approx \phi_*$. Hence, 
from (\ref{eq:thm1}),  
$$\phi_*^2 \approx \frac{1}{2} (1 - \epsilon)$$
and thus, 
$$c = 2 \sin(\phi_*) \approx 2 \phi_* = \sqrt{2 (1 - \epsilon)}$$ 
which coincides with the ZFK value. 

\section{Relaxed Problem}
In this section we prove the existence of a maximizer for the functional
\begin{equation}
\mathcal{F}(u) = \frac{\frac{1-\varepsilon^2}{s_0}+ \int_0^{s_0} \frac{[u^2-\varepsilon^2]_+}{s^2} ds}{\int_0^{s_0} u'(s)^2 ds} \ .
\end{equation}
This functional corresponds to (\ref{newvp}) where the reaction term has been replaced by $f_L$. 
It has to be maximized over all positive increasing functions $u(s)$ suject to the conditions $u(0)=0$ and $u(s_0)=1$.

Consider the functional
\begin{equation}
\mathcal{G}(u) = \frac{1}{2}\frac{\int_0^\infty \frac{[u^2-\varepsilon^2]_+}{s^2} ds}{\int_0^\infty u'(s)^2 ds}
\label{fnalG}
\end{equation}
for functions $u: \R_+ \rightarrow [0,1]$ with
$u(0)=0$, $u$ increasing and $\lim_{s \to \infty} u(s) = 1$, i.e., we do not reqire that $u(s)$ assumes the value $1$ at some
point $s_0$. We denote this set of functions by $\mathcal{C}$.
For any positive, increasing function $u$ with $u(0)=0$ and $u(s_0)=1$ for some $s_0$ we set 
$v(s) := u(s)$ for $s \le s_0$ and $v(s) \equiv 1$ for $s \ge s_0$.  The function $v(s)$ is in $\mathcal{C}$ and a simple computation shows that
\begin{equation}
\mathcal{F}(u) = 2\mathcal{G}(v) \ .
\end{equation}
Thus, the supremum of $\mathcal{F}(u)$ over all functions $u$ and all values of $s_0$ is not larger than the supremum of $2\mathcal{G}(u)$.
\begin{lm} [ZFK bound]\label{bound}
The functional $\mathcal{G}$ is bounded above, in fact
\begin{equation} \label{ZFKbound}
\mathcal{G}(u) \le 2 \frac{1-\varepsilon^2}{(1+\varepsilon^2)^2}
\end{equation}
\end{lm}
\begin{remark}
Note that as $\varepsilon$ gets close to $1$ the right side of (\ref{ZFKbound}) tends to  $\frac{1-\varepsilon}{2}$
and hence the supremum of $\sqrt{\mathcal{F}(u)}$ which corresponds to the wave speed, is less than $ \sqrt{2(1-\varepsilon)}$ as $\varepsilon \to 1$.
\end{remark}
\begin{proof} Fix any function $u$ in $\mathcal{C}$ with both, the numerator and the denominator of $\mathcal{G}$ finite. Since the derivative of $u$ is square integrable, the function $u$ is H\"older continuous. In fact by an elementary estimate
\begin{equation}\label{hoelder}
|u(s)-u(s')| \le \sqrt{|s-s'|} \sqrt \Vert u' \Vert_2 \ .
\end{equation}

Using the scaling invariance of $\mathcal{G}$ we may assume that
$u(\varepsilon)=\varepsilon $.
The numerator $\mathcal{N}(u)$ of $\mathcal{G}(u)$ can be written as
\begin{equation}
\mathcal{N}(u)=\int_\varepsilon ^\infty \frac{u^2-\varepsilon^2}{s^2} ds
= \int_\varepsilon ^\infty \frac{u^2}{s^2} ds - \varepsilon
= 2\int_\varepsilon ^\infty \frac{uu'}{s} ds \ .
\end{equation}
Now by Schwarz's inequality
\begin{equation}
\mathcal{N}(u) \le 2(\int_\varepsilon ^\infty \frac{u^2}{s^2} ds)^{1/2} (\int_\varepsilon ^\infty u'^2 ds)^{1/2}
\end{equation}
or
\begin{equation}
\frac{1}{4}\frac{\mathcal{N}(u)^2}{\mathcal{N}(u) +\varepsilon} \le \int_\varepsilon ^\infty u'^2 ds \ .
\end{equation}
Since $u(0)=0, u(\varepsilon) = \varepsilon$ and since the numerator
does not depend on $u$ on this interval $[0, \varepsilon]$ we find that the denominator of $\mathcal{G}(u)$ is smallest when $u(s)=s$ and hence it is bounded below by
\begin{equation} \label{denombound}
\varepsilon +  \int_\varepsilon ^\infty u'^2 ds \ .
\end{equation}
Thus
\begin{equation}
\mathcal{G}(u) \le 2 \frac{\mathcal{N}(u)^2 +\varepsilon \mathcal{N}(u)}{\mathcal{N}(u)^2 +4 \varepsilon \mathcal{N}(u) +4\varepsilon^2}
\end{equation}
Note that the right side as a function of $\mathcal{N}(u)$ is increasing. Further, since
\begin{equation}\label{numbound}
\mathcal{N}(u) \le \frac{1-\varepsilon^2}{\varepsilon}
\end{equation}
the stated estimate is established.
\end{proof}

\begin{lm}  \label{existence}
There exists a function $u \in \mathcal{C}$ such that
\begin{equation}
\mathcal{G}(u) = \sup_{v \in \mathcal{C}}\mathcal{G}(v) =: M\ .
\end{equation}
Moreover, if we normalize $u$ so that
$u(\varepsilon)=\varepsilon$ then on the interval $(0,\varepsilon)$,
$u(s)=s$.
\end{lm}
\begin{proof}Let $u_n$ be a maximizing sequence, i.e.,
\begin{equation}
\mathcal{G}(u_n) \to M \ .
\end{equation}
By scaling we can assume that $u_n(\varepsilon)=\varepsilon $.
By (\ref{numbound}) the numerator $\mathcal{N}(u_n)$ is bounded.
Hence the denominator is also bounded. Thus
\begin{equation}
\int_0 ^\infty u_n'^2 ds \le C
\end{equation}
for some constant $C$ independent of $n$. By (\ref{hoelder}) the functions $u_n$ are uniformly continuous. Since the functons $u_n$ are uniformly bounded, by Arzela-Ascoli we can
pass to a subsequence, again denoted by $u_n$ which converges uniformly on any finite interval of $[0, \infty)$ to some function
$u$. This function is in $\mathcal{C}$ since the pointwise limit of monotone functions is monotone.
Since $u_n$ is bounded it follows from the dominated convergence theorem that
\begin{equation}
\lim_{n \to \infty} \int_\varepsilon^\infty \frac{u_n^2}{s^2} ds =
\int_\varepsilon^\infty \frac{u^2}{s^2} ds \ ,
\end{equation}
and
\begin{equation}
\liminf_{n \to \infty} \int_0^\infty u_n'^2 ds \ge \int_0^\infty u'^2 ds
\end{equation}
by the weak lower semicontinuity of the $L^2$-norm. Thus
\begin{equation}
M = \lim_{n \to \infty} \mathcal{G}(u_n) \le \mathcal{G}(u)
\end{equation}
and hence $\mathcal{G}(u) = M$.
Note that both, the denominator and the numerator of $\mathcal{G}(u)$ are finite.
Assume that
\begin{equation}
\lim_{s \to \infty}u(s) = a < 1 \ .
\end{equation}
Then $u/a \in \mathcal{C}$ and
\begin{equation}
\mathcal{G}(\frac{u}{a})= \frac{1}{2}\frac{\int_0^\infty \frac{[u^2-(\varepsilon a)^2]_+}{s^2} ds}{\int_0^\infty u'(s)^2 ds} > \mathcal{G}(u) \ ,
\end{equation}
and hence $a=1$. Thus $u \in \mathcal{C}$ and $u$ is a maximizer.
\end{proof}
It remains to analyze the maximizers $u$ of the functional $\mathcal{G}(u)$. To this end we relax the functional once more.
Let $v$ be a function with $0 \le v(s) \le 1, v(0)=0,$ and monotone
{\it except} on some open subinterval of $[0, \infty)$. We denote this domain by $\mathcal{D}$. Consider the function
\begin{equation}
u_v(s) = \int_0^s \max\{v'(t), 0\}dt \ .
\end{equation}
Clearly, this function is monotone and $u_v \ge v$ pointwise.
Clearly, $\mathcal{N}(u_v) \ge \mathcal{N}(v)$ and $\Vert u_v' \Vert_2
\le \Vert v' \Vert_2$.
Thus
\begin{equation}
\mathcal{G}(v) \le \mathcal{G}(u_v) \ ,
\end{equation}
and the maximizer found before, maximizes the functional in the larger class $\mathcal{D}$.
\begin{lm} \label{neumann}
let $u$ be a maximizer of the functional $\mathcal{G}(u)$. Then
\begin{equation}
s_0 = \inf\{s: u(s) = 1\}
\end{equation}
is finite and $u'(s_0) = 0$. Moreover, on the interval $(0,\varepsilon]$, $u(s) = s$, on the interval $(\varepsilon, s_0]$
the function $u(s)$ is of the form
\begin{equation}
u(s) = \sqrt s A  \cos(\frac{1}{2}\sqrt{2/M - 1})\log s + \delta)
\end{equation}
for suitable constants $A$ and $\delta$. Finally, on $(s_0, \infty)$,
$u(s) \equiv 1$.

\end{lm}
\begin{proof}
Assume that $s_0 =\infty$. By scaling we can assume that $u(\varepsilon)=\varepsilon$.
Moreover since the numerator of $\mathcal{G}$ does not depend on the function on the interval $[0,\varepsilon]$ the denominator is smallest by choosing $u(s)=s$ in that interval. In particular
we have that $u(s) \ge \varepsilon$ on $[\varepsilon, \infty)$.  Further $u(s) < 1 $ for all $s$.
Pick any smooth function $f$ with compact support in $(\varepsilon ,\infty)$. Then for $t$ small enough $u+tf \in \mathcal{D}$
and
\begin{equation}\label{varequation}
\int \frac{u}{s^2}f - 2M \int u'f' = 0
\end{equation}
i.e.,
\begin{equation}
\frac{u}{s^2}+2Mu'' = 0
\end{equation}
in the weak sense on the intervall $(\varepsilon, \infty)$.
Any weak solution of this equation
is a linear combination of $s^{\alpha_{\pm}}$
where $\alpha_{\pm} = \frac{1}{2}(1 \pm \sqrt{1 -2/M})$.
By Lemma \ref{bound} $M < 2$ and we get two complex conjugate roots. Hence
the real solutions are
\begin{equation} \label{solution}
u(s) = \sqrt s A  \cos(\frac{1}{2}\sqrt{2/M - 1})\log s + \delta) \ ,
\end{equation}
where $A$ and $\delta$ are constants. The form of $u(s)$ contradicts the monotonicity of $u(s)$ and hence $s_0 < \infty$. This establishes
also the form of $u(s)$ on $(\varepsilon, s_0]$.
To see that $u'(s_0)=0$ pick any nonnegative smooth function $f$,
whose support is in a close vicinity of $s_0$.
Clearly $u-tf$ is in $\mathcal{D}$ for $t>0$ and small.
Then
\begin{equation}
\mathcal{G}(u-tf) \le \mathcal{G}(u)
\end{equation}
and hence
\begin{equation} \label{intinequality}
\int \frac{u}{s^2}f - 2M \int u'f' \ge 0
\end{equation}
On the interval $(\varepsilon, s_0)$ the function $u$ satisfies the
equation (\ref{varequation}) and hence the right side of
(\ref{intinequality}) can be rewritten as
\begin{equation}
 \int^{s_0}( \frac{u}{s^2}f - 2M u'f') ds+ \int_{s_0}\frac{1}{s^2}f ds
\end{equation}
which, using integration by parts, yields
\begin{equation}
0 \le -2Mu'(s_0)f(s_0) + \int_{s_0}\frac{1}{s^2}f(s) ds \ .
\end{equation}
Since $u$ is increasing we have
\begin{equation}
0 \le 2Mu'(s_0) \le \frac{1}{f(s_0)}\int_{s_0}\frac{1}{s^2}f(s) ds \ .
\end{equation}
Assume further that $f$ is  a non-negative, smooth, compactly supported function with $f(s_0)=1$
and set $f_n(s)= f((s-s_0)n+s_0)$. A simple calculation shows that
as $n$ tends to infinity
\begin{equation}
\frac{1}{f(s_0)}\int_{s_0}\frac{1}{s^2}f(s) ds \approx \frac{1}{n s_0^2}
\int_{s_0}^\infty f(s)ds \ .
\end{equation}
Hence $u'(s_0)=0$.
\end{proof}

\section{The maximizer}

In this section we determine explicitly the optimizer, whose existence was established in the
previous section.

\begin{thm} \label{maximizer}
The unique maximizer is given by
\begin{equation}
u(s) =   \left \{\begin{array}{l}
s  \  {\rm if} \ 0 \le s \le \varepsilon \\
 A \sqrt s \cos(\phi(s))  \  {\rm if} \  \varepsilon < s < s_0 \end{array} \right.
 \end{equation}
with
\begin{equation}
A = \frac{ \sqrt \varepsilon }{\cos(\phi_{*})} \ \ \ s_0= \frac{1}{\varepsilon}
\end{equation}
and
\begin{equation}
 \phi(s) = \frac{1}{2} \cot(\phi_{*}) \log(\frac{s}{\varepsilon}) -\phi_{*} \ .
\label{eq:39}
\end{equation}
Here $\phi_{*} $ is the first positive solution of the equation
\begin{equation} \label{tangentequation}
\phi_{*} \tan(\phi_{*}) = \frac{1}{2} |\log(\varepsilon)| \ .
\end{equation}
Moreover, we have 
\begin{equation}
M=: \sup_{v \in \mathcal{C}}\mathcal{G}(v) = \mathcal{G}(u)= 2 \sin^2 \phi_{*}.
\label{valueofM}
\end{equation}

\end{thm}

\begin{proof} In a first step we show that
\begin{equation} \label{essnaught}
s_0= \frac{1}{\varepsilon} \ .
\end{equation}
We know by Lemma \ref{existence} and Lemma \ref{neumann} that there exists a maximizer
with the following properties
\begin{equation}
u(\varepsilon) = \varepsilon \ \ \ u(s_0) = 1 \ \ \ u'(s_0) = 0 \ .    \label{propone}
\end{equation}
Since $u'(s)$ is continuous and $u(s) = s$ for $s \le \varepsilon$ we also have
\begin{equation}
u'(\varepsilon) = 1 \ . \label{propotwo}
\end{equation}
Moreover, on the interval $[\varepsilon, s_0]$ the function $u(s)$ is positive and increasing and has the form
\begin{equation}
u(s) = \sqrt s A  \cos\left(\frac{1}{2}\sqrt{\frac{2}{M} - 1}\log s + \delta\right) \label{propthree}
\end{equation}
where $M = \mathcal{G}(u)$, the maximal value of the functional. Note that by Lemma \ref{bound} $M < 2$.

Since $u(\varepsilon) = \varepsilon$ and $u'(\varepsilon) = 1$ we have, using (\ref{propthree})
\begin{equation} \label{one}
\sqrt \varepsilon = A \cos\left(\frac{1}{2} \sqrt{\frac{2}{M}-1}\log \varepsilon + \delta\right)
\end{equation}
\begin{equation} \label{two}
-\sqrt \varepsilon = A \sqrt{\frac{2}{M}-1} \sin\left(\frac{1}{2} \sqrt{\frac{2}{M}-1}\log \varepsilon + \delta\right) \ .
\end{equation}
Similarly, from the fact that $u(s_0)=1$ and $u'(s_0) = 0$ we get from ( \ref{propthree})
\begin{equation}  \label{three}
1= \sqrt {s_0} A  \cos(\frac{1}{2}\sqrt {\frac{2}{M} - 1}\log s_0 + \delta)
\end{equation}
\begin{equation} \label{four}
1= \sqrt {s_0} A \sqrt{\frac{2}{M} - 1}  \sin(\frac{1}{2}\sqrt{\frac{2}{M} - 1}\log s_0 + \delta) \ .
\end{equation}

Next we prove (\ref{tangentequation}), i.e., we calculate $\mathcal{G}(u) = M$. A straightforward calculation
yields for the numerator
\begin{eqnarray}
& &\int_0^\infty \frac{[u(s)^2 - \varepsilon^2]}{s^2} ds  \\
&=&\frac{A^2}{2} \left( \log\frac{s_0}{\varepsilon} +  \frac{1}{ \sqrt {\frac{2}{M} - 1}}
( \sin(\sqrt{\frac{2}{M} - 1}\log s_0 + 2 \delta)  -
\sin(\sqrt{\frac{2}{M}-1}\log \varepsilon + 2\delta) \right) \ .
\end{eqnarray}
Likewise, for the denominator
\begin{eqnarray}
& &2 \int_0^\infty u'(s)^2 ds  = 2 \varepsilon + \frac{A^2}{2M} \log\frac{s_0}{\varepsilon} \\
&+&\frac{A^2}{2} (1-\frac{1}{M}) \frac{1}{\sqrt{\frac{2}{M} - 1}}\left(( \sin(\sqrt{\frac{2}{M} - 1}\log s_0 + 2 \delta)  -
\sin(\sqrt{\frac{2}{M}-1}\log \varepsilon + 2\delta)   \right) \\
&+&\frac{A^2}{2} \left( \cos(\sqrt{\frac{2}{M} - 1}\log s_0 + 2 \delta)  -
\cos(\sqrt{\frac{2}{M}-1}\log \varepsilon + 2\delta)   \right) \ .
\end{eqnarray}

The equation
$$
\int_0^\infty \frac{[u(s)^2 - \varepsilon^2]}{s^2} ds = M \left(2 \int_0^\infty u'(s)^2 ds\right)
$$
then reduces to
\begin{eqnarray}
0&= &2 \varepsilon  -  \frac{A^2}{2} \sqrt{\frac{2}{M} - 1}
\left( \sin(\sqrt{\frac{2}{M} - 1}\log s_0 + 2 \delta)  -
\sin(\sqrt{\frac{2}{M}-1}\log \varepsilon + 2\delta)   \right) \\
&+& \frac{A^2}{2} \left( \cos(\sqrt{\frac{2}{M} - 1}\log s_0 + 2 \delta)  -
\cos(\sqrt{\frac{2}{M}-1}\log \varepsilon + 2\delta)   \right)
\end{eqnarray}
Using (\ref{one}-\ref{four}) together with the double angle formulas for cosine and sine one easily sees that
the above equation reduces to
\begin{equation}
\left(\varepsilon - \frac{1}{s_0}\right) \left(\frac{2}{2-M}\right) = 0 \ ,
\end{equation}
and hence (\ref{essnaught}) is proved.

The next step is to calculate $M$. Note that (\ref{three}) and (\ref{four}) now read
\begin{equation}  \label{threeprime}
1=\frac{1}{\sqrt \varepsilon}A  \cos(-\frac{1}{2}\sqrt {\frac{2}{M} - 1}\log \varepsilon + \delta)
\end{equation}
\begin{equation} \label{fourprime}
1= \frac{1}{\sqrt \varepsilon} A \sqrt{\frac{2}{M} - 1}  \sin(-\frac{1}{2}\sqrt{\frac{2}{M} - 1}\log \varepsilon + \delta) \ .
\end{equation}
from which we deduce that
\begin{equation} \label{tangentone}
\tan(\frac{1}{2}\sqrt{\frac{2}{M} - 1}\log \varepsilon - \delta)  = -\frac{1}{\sqrt{\frac{2}{M} - 1}} \ .
\end{equation}
Likewise from (\ref{one}) and (\ref{two}) we obtain
\begin{equation} \label{tangenttwo}
\tan(\frac{1}{2}\sqrt{\frac{2}{M} - 1}\log \varepsilon + \delta)  = -\frac{1}{\sqrt{\frac{2}{M} - 1}} \ .
\end{equation}
Using the addition formula for the tangent function yields
\begin{equation}  \label{Mequation}
\tan(\sqrt{\frac{2}{M} - 1}\log \varepsilon) = -\tan(\sqrt{\frac{2}{M} - 1}|\log \varepsilon|)
= -\frac{\sqrt{\frac{2}{M} - 1}}{\frac{1}{M} -1} \ .
\end{equation}
If we set
\begin{equation}
\phi_{*} = \frac{1}{2}\sqrt{\frac{2}{M} - 1}|\log \varepsilon|
\end{equation}
and note that
\begin{equation}
\tan(\sqrt{\frac{2}{M} - 1}|\log \varepsilon|) = \frac{2 \tan \phi_{*}}{1 - (\tan \phi_{*})^2} \ ,
\end{equation}
we learn that
\begin{equation}
\tan \phi_{*} = \frac{1}{\sqrt{\frac{2}{M} - 1}} = \frac{|\log \varepsilon|}{2 \phi_{*}}
\label{eq:a1}
\end{equation}
which yields (\ref{tangentequation}).
Since $\phi_{*}  >0$ and $M$ is the maximum of our functional, we have to choose $\phi_{*} $ to be
the first positive solution of (\ref{tangentequation}). In particular we have that $\phi_{*} < \pi/2$.

It remains to determine $\delta$ and $A$. Subtracting (\ref{tangentone}) from (\ref{tangenttwo}) we find that
\begin{equation}
\tan(2 \delta) = 0
\end{equation} and hence $ \delta = N \pi /2$ where $N \in \mathbb{Z}$. Note that as $s$ ranges from $\varepsilon$
to $1/\varepsilon$, the function
$$
\frac{1}{2}\sqrt{\frac{2}{M} - 1}\log s
$$
varies from $-\phi_{*}$ to $\phi_{*}$. The function $u(s)$ is positive and increasing and hence,
if we choose the constant $A$ positive, we find that $\delta = 2\pi N$ where $N \in \mathbb{Z}$.
Hence we may choose $\delta= 0$. The function $\frac{1}{2}\sqrt{\frac{2}{M} - 1}\log s$ can be conveniently
be written as
\begin{equation}
\frac{1}{2}\sqrt{\frac{2}{M} - 1}\log s = \frac{1}{2} \cot{\phi_{*}} \log\frac{s}{\varepsilon} -\phi_{*}
\end{equation}
and the condition that $u(\varepsilon)=\varepsilon$ yields the value for the constant $A$ stated
in Theorem \ref{maximizer}.
Finally, equation (\ref{valueofM}) for the value of $M$ follows immediately from the first equality in (\ref{eq:a1}).
\end{proof}

\section{Error Estimates: Proof of Theorem \ref{main}} 

In section 3 we have determined the exact value of the propagation speed, $c_L$ say, of initially localized disturbances 
obeying (\ref{rd}) for a linear profile with a cutoff. In fact, if the profile is given piecewise by $f(u)=0$, for  
$u < \epsilon$,  and $f(u)=u$ for $\epsilon \le u \le 1$, we have shown that $c_L$ is given exactly by
\begin{equation}
c_L=2 \sin \phi_{*}
\label{eq:4.1}
\end{equation}
where $\phi_{*}$ is the first positive solution of the equation
\begin{equation}
\phi_{*} \tan \phi_* = \frac{1}{2} |\log \epsilon |.
\label{eq:4.2}
\end{equation}
Solving (\ref{eq:4.2}) for $\phi_*$ in power series on $1/|\log \epsilon|$, and replacing it in (\ref{eq:4.1}) 
one finds that
$$
c_L = 2 - \frac{\pi^2}{|\log \epsilon |^2} + \rm{o} \left(\frac{1}{|\log \epsilon |^2} \right)
$$
where the leading two terms account precisely for the Brunet and Derrida result (i.e., Equation {\ref{brunet} in the Introduction).  

Here, we would like to determine error bounds when the profile $f(u)$ is a KPP profile with a cutoff, in other words, when the profile is given piecewise by $f(u)=0$, for  
$0 \le u < \epsilon$,  and $f(u) \le u$ for $\epsilon \le u \le 1$. If we write $f(u) = u -N(u)$, for $\epsilon \le u \le 1$, the KPP criterion amounts to requiring that $N(u) \ge 0$. For such a reaction profile, we have that
$$
F(u) \equiv  \int_0^u  f(q) \, dq, 
$$
is such that $F(u)=0$ for $0 \le u \le \epsilon$, whereas 
$$
F(u) = \frac{1}{2}(u^2 - \epsilon^2) - \int_{\epsilon}^u N(q) \, dq,
$$
for $\epsilon \le u \le 1$. For a KPP profile $N(u) \ge 0$, thus, 
\begin{equation}
F(u) \le G(u) 
\label{eq:FG}
\end{equation}
where 
$$
G(u) = \frac{1}{2} \left(u^2-\epsilon^2 \right)_+.
$$
Hence, using (\ref{eq:FG}) in (\ref{newvp}), and,  taking into account (\ref{fnalG}) and (\ref{valueofM}), we see that in general for a KPP profile with a cutoff, the speed of 
propagation of fronts for an initially localized disturbance of (\ref{rd}), say $c$, satisfies,
$$
c \le c_L.
$$
On the other hand, we can also use the variational principle embodied in (\ref{newvp}) to obtain 
a lower bound on $c$. For that purpose we use as a trial function in (\ref{newvp}) the minimizer
$\hat u$ of the functional $G$. After some simple computations, we obtain,
\begin{equation}
c_L^2-c^2 \le 
\frac{\int_{\epsilon}^{1/\epsilon} N(\hat u(s)) {\hat u}'(s) (1/s)  \, ds}
{\int_{0}^{1/\epsilon} ({\hat u}'(s))^2 \, ds}.
\label{eq:4.4}
\end{equation}
Here, we will find estimates on the difference $c_L^2-c^2$ for profiles that satisfy the bound,
\begin{equation}
0\le N(x) \le B (x-\epsilon)^{1+\eta} 
\label{eq:4.5}
\end{equation}
for $\epsilon \le x \le 1$, where $\eta >1$. The denominator can be calculated in closed form as follows,
\begin{eqnarray}
D &= {\int_{0}^{1/\epsilon} ({\hat u}'(s))^2 \, ds}=\epsilon+{\int_{\epsilon}^{1/\epsilon} ({\hat u}'(s))^2 \, ds} \nonumber\\
&=\epsilon + \frac{\epsilon}{4 \cos^2 \phi_* \sin^2 \phi_*} \frac{\sin \phi_*}{\cos \phi_*}
\int_{-\phi_*}^{\phi_*} 2 \sin(\phi_*-t)^2 \, dt \nonumber \\
&=\epsilon + \frac{\epsilon}{4 \cos^2 \phi_* \sin^2 \phi_*} \frac{\sin \phi_*}{\cos \phi_*}
\int_{-\phi_*}^{\phi_*} \left(1 - \cos(2 \phi_* -2t) \right) \, dt \nonumber\\&=\epsilon + \frac{\epsilon}{4 \cos^2 \phi_* \sin^2 \phi_*} \frac{\sin \phi_*}{\cos \phi_*}\left(2 \phi_* - (1/2) \sin 4 \phi_* \right) \nonumber\\
&=\epsilon \frac{1}{4 \cos^3\phi_* \sin \phi_*}\left(2 \phi_* + \sin 2 \phi_* \right).
\label{eq:den}
\end{eqnarray} 
On the other hand, using (\ref{eq:4.5}) in the numerator of (\ref{eq:4.4}), the properties of the trial function $\hat u(s)$ (in particular the fact that this function is increasing), we can estimate the numerator as, 
\begin{equation}
N \le \frac{B \sqrt{\epsilon}}{2 \cos \phi_* \sin \phi_*} \int_{\epsilon}^{1/\epsilon} \left(\sqrt{\epsilon s} \frac{\cos \phi}{\cos \phi_*} - \epsilon \right)^{1+\eta} \sin(\phi_* - \phi) \, \frac{ds}{s^{3/2}}.
\label{eq:4.6}
\end{equation}
Using the fact that $s=\exp{(2 \phi \tan \phi_*)}$ (which follows from (\ref{eq:39}) and (\ref{tangentequation}) above) and that we can write $\sqrt{\epsilon} =\exp{((\log \epsilon)/2)}=
\exp{(-|\log \epsilon|/2)}= \exp{(-\phi_* \tan \phi_*)}$, we have that
$$
\sqrt{s \epsilon} = \exp{[(\phi-\phi_*) \tan \phi_*]},
$$
and also that
$$
\frac{ds}{s^{3/2}} = 2 \sqrt{\epsilon}\tan \phi_* \exp{((\phi_*-\phi)\tan \phi_*)}.
$$
Changing the variable of integration from $s$ to $\phi$ in (\ref{eq:4.6}), making use of these last two expressions, we find,
\begin{equation}
N \le \frac{B {\epsilon}}{\cos^2 \phi_*} \int_{-\phi_*}^{\phi_*}
\left(\exp{[-(\phi_*-\phi) \tan \phi_*]}\frac{\cos \phi}{\cos \phi_*} - \epsilon \right)^{1+\eta} \sin(\phi_* - \phi) 
\exp{(\phi_*-\phi) \tan \phi_*)}
\, d \phi.
\label{eq:4.7}
\end{equation}
Finally making the change of variables $\phi \to \sigma=\phi_*-\phi$ we get, 
\begin{equation}
N \le \frac{B {\epsilon}}{\cos^2 \phi_*} \int_{0}^{2\phi_*}
\left(\exp{[-\sigma \tan \phi_*]}\frac{\cos (\phi_*-\sigma)}{\cos \phi_*} - \epsilon \right)^{1+\eta} 
\sin \sigma  
\exp{(\sigma \tan \phi_*)}
\, d \sigma.
\label{eq:4.8}
\end{equation}
Hence, from (\ref{eq:den}) and (\ref{eq:4.8}), we have, 
\begin{equation}
\frac{N}{D} 
\le 4 \, B \frac{\cos \phi_* \sin \phi_*}{2 \phi_* + \sin(2\phi_*)} I, 
\label{eq:4.9}
\end{equation}
with, 
\begin{equation}
I= \int_{0}^{2\phi_*}
\left(\exp{[-\sigma \tan \phi_*]}\frac{\cos (\phi_*-\sigma)}{\cos \phi_*} - \epsilon \right)^{1+\eta} 
\sin \sigma  
\exp{(\sigma \tan \phi_*)}
\, d \sigma.
\label{eq:4.10}
\end{equation}
When $\epsilon \to 0$, we have from (\ref{tangentequation}) that $\phi_* \approx \pi/2$, $\sin \phi_* \approx 1$, $\sin (2 \phi_*) \approx 0$ and $\cos \phi_* =O\left( 1/|\log \epsilon| \right)$. Thus, in order to control the difference $c_L^2-c^2$, all we have to prove is that
$$
I \le o \left( 1/|\log \epsilon| \right).
$$
We can estimate $I$ from above by dropping the $\epsilon$ inside the factor in the integral above. Moreover, we 
write $\cos(\phi_*-\sigma)/\cos \phi_* = \cos \sigma+\tan \phi_* \sin \sigma$. Thus, we have,
\begin{equation}
I \le J \equiv \int_{0}^{2\phi_*}
\left( \cos \sigma+\tan \phi_* \sin \sigma\right)^{1+\eta} 
\sin \sigma  
\exp{(-\sigma \eta \tan \phi_*)}.
\, d \sigma.
\label{eq:4.11}
\end{equation}
We now split the integral over $\sigma$ into two parts. We denote by $J_1$ the integral between $0$ and $\alpha$ and 
by $J_2$ the integral between $\alpha$ and $2\phi_*$. The value of $\alpha$ will be conveniently chosen later on. We will first estimate $J_1$. We use: i) $\exp(-\sigma \eta \tan  \phi_*) \le 1$ (since $\tan \phi_*>0$); ii) $\cos \sigma \le 1$ and iii) $\sin \sigma \le \sigma \le \alpha$ to get
\begin{equation}
J_1 \le \int_0^{\alpha} \left(1+ \alpha \tan \phi_*  \right)^{1+\eta} \sigma \, d\sigma = \frac{\alpha^2}{2}\left(1+ \alpha \tan \phi_* \right)^{1+\eta},
\label{eq:4.12}
\end{equation}
and, using the convexity of $x \to x^{1+\eta}$ (since $\eta >0$), we have,
\begin{equation}
J_1 \le 2^{\eta-1} \left[\alpha^2 + \alpha^{3+\eta} (\tan \phi_*)^{1+\eta}\right].
\label{eq:4.13}
\end{equation}
On the other hand, in order to estimate $J_2$, we use the fact that $0 \le \cos \sigma, \sin \sigma \le 1$ in the interval $[0,\phi_*]$ (recall that $\phi_* \le \pi/2$). We also use  that $\exp(-x)$ is decreasing, and we get at once
\begin{equation}
J_2 \le \left(1+\tan \phi_*\right)^{1+\eta}\exp(-\alpha \eta \tan \phi_*) \, 2 \phi_*,
\label{eq:4.14}
\end{equation}
Pick any  $0<r<1$, and then choose $\alpha$ to be, 
\begin{equation}
\alpha=(\tan \phi_*)^{-(2+\eta+r)/(3+\eta)}
\label{eq:4.15}
\end{equation}
The idea behind this choice is that it will make $J_1=o\left(1/|\log \epsilon| \right)$, and at the same time 
it will make $J_2$ of smaller order. 
Now, one can easily check that, for any $0<r<1$, 
\begin{equation}
2\left(\frac{2+\eta+r}{3+\eta}\right) > 1+r.
\label{eq:4.16}
\end{equation}
Now, since for $\epsilon$ small, $\tan \phi_* >1$, it follows from (\ref{eq:4.13}), our choice of $\alpha$ (i.e., equation (\ref{eq:4.15})), and (\ref{eq:4.16}) that, 
\begin{equation}
J_1 \le 2^{\eta} (\tan \phi_*)^{-(1+r)}.
\label{eq:4.17}
\end{equation}
Finally, using (\ref{tangentequation}) in (\ref{eq:4.17}), and the fact that $r>0$, we get the desired estimate,
\begin{equation}
J_1 = o\left(\frac{1}{|\log \epsilon |} \right).
\label{eq:4.18}
\end{equation}
Also, with our choice of $\alpha$, (\ref{eq:4.15}), 
$$
\alpha \tan \phi_* = \left(\tan \phi_* \right)^{(1-r)/(3+\eta)},
$$
where $r>1$ and $\eta>0$. Using (\ref{eq:4.14}), we see that $J_2$ is exponentially small as a function of $1/|\log \epsilon|$ (the first factor grows polynomially as a function of $\tan \phi_*$, while the second factor is exponentially small). 
Summarizing, we have proven that
$$
0 \le c_L^2-c^2 \le o \left(\frac{1}{|\log \epsilon |^2} \right)
$$

\section{Appendix}
For the sake of completeness, in this appendix we prove bounds on $c_L^2-c^2$ in terms of the parameter $\eta$, for KPP profiles. These bounds allow us to show that as $\eta \to \infty$, $c_L^2-c^2 \to 0$.
Consider, 
\begin{equation}
J=\int_0^{2 \phi_*} \left(\cos \sigma + \tan \phi_* \sin \sigma \right)^{1+\eta} \sin \sigma \exp{(-\sigma \eta \tan \phi_*)} \, d\sigma.
\label{eq:A1}
\end{equation}Denote by 
\begin{equation}
H \equiv = \sigma \tan \phi_* - \log(\cos \sigma + \tan \phi_* \sin \sigma),
\label{eq:A2}
\end{equation}
and notice that
$$
H_{\sigma} \equiv \frac{dH}{d \sigma}= \sin \sigma \frac{(1+\tan^2 \phi_*)}{(\cos \sigma + \tan \phi_* \sin \sigma)}>0.
$$
Using (\ref{eq:A1}) and (\ref{eq:A2}), we can write,
\begin{equation}
J=\int_0^{2 \phi_*} \left(\cos \sigma + \tan \phi_* \sin \sigma \right) \exp{(- \eta H)} \, \sin \sigma  d\sigma.
\label{eq:A3}
\end{equation}
which can be rewriten as,
\begin{equation}
J=\frac{1}{1+\tan \phi_*^2}\int_0^{2 \phi_*} \left(\cos \sigma + \tan \phi_* \sin \sigma \right)^2 \exp{(- \eta H)} \, H_{\sigma}  d\sigma.
\label{eq:A4}
\end{equation}
We have reamarked before that on the interval $(0,2\phi_*)$, both $0<\cos \sigma, \sin \sigma <1$, thus
$(\cos \sigma + \tan \phi_* \sin \sigma)^2 \le (1+\tan \phi_*)^2$. Moreover, using that $(1+x)^2/(1+x^2) \le 2$ for $x \ge 0$, we can finally write, 
$$
J \le 2 \int_0^{2 \phi_*} e^{-H \eta} H_{\sigma} \, d \sigma.
$$
Recalling that $H_{\sigma}>0$, making the change of variables $\phi \to H$, and computing $H(0)=0$ and $H(2\phi_*)=2\phi_* \tan \phi_*$ we get, 
\begin{equation}
J \le \int_0^{2 \phi_* \tan \phi_*} e^{-H \eta} \, d H \le \frac{2}{\eta} 
\left(1-\exp{(-2 \eta \phi_*\tan \phi_*)} \right)  \le \frac{2}{\eta}
\label{eq_A6}
\end{equation}
Hence, $J \to 0$ as $\eta \to \infty$.

\end{document}